\documentclass[manuscript, nonacm=true]{acmart}

\usepackage{verbatim}

\newcommand{\Gr}{\mathcal{G}}
\newcommand{\N}{\mathcal{N}}
\newcommand{\A}{\mathbf{A}}
\newcommand{\vv}{\mathbf{v}}
\usepackage{amsmath}
\usepackage{graphicx}
\usepackage{xcolor}
\definecolor{blue}{rgb}{0,0,0}
\long\def\COMMENT#1\ENDCOMMENT{\message{(Commented text...)}\par}
\usepackage{hyperref}
\usepackage{url}

\AtBeginDocument{%
  \providecommand\BibTeX{{%
    \normalfont B\kern-0.5em{\scshape i\kern-0.25em b}\kern-0.8em\TeX}}}

\setcopyright{acmcopyright}
\copyrightyear{2019}
\acmYear{2019}
\acmConference[WI '19]{IEEE/WIC/ACM International Conference on Web Intelligence}{October 14--17, 2019}{Thessaloniki, Greece}
\acmBooktitle{IEEE/WIC/ACM International Conference on Web Intelligence (WI '19), October 14--17, 2019, Thessaloniki, Greece}
\acmPrice{15.00}
\acmDOI{10.1145/3350546.3352559}
\acmISBN{978-1-4503-6934-3/19/10}

\begin{document}

\title{Potential gain as a centrality measure}


\author{Pasquale De Meo}
\affiliation{%
  \institution{Dept. of Ancient and Modern Civilizations\\ University of Messina}
  \streetaddress{}
  \city{I-98168 Messina}
  \country{Italy}}
\email{pdemeo@unime.it}

\author{Mark Levene}
\email{mark@dcs.bbk.ac.uk}
\orcid{}
\author{Alessandro Provetti}
\email{ale@dcs.bbk.ac.uk}
\affiliation{%
  \institution{Dept. of Computer Science \\ Birkbeck, University of London}
  \streetaddress{}
  \city{London WC1E 7HX}
  \state{UK}
  \postcode{}
}

\renewcommand{\shortauthors}{De Meo, et al.}

\begin{abstract}
Navigability is a distinctive features of graphs associated with artificial or natural systems whose primary goal is the transportation of information or goods. 
We say that a graph $\Gr$ is navigable when an agent is able to efficiently reach any target node in $\Gr$ by means of local routing decisions. 
In a social network navigability translates to the ability of reaching an individual through personal contacts.
Graph navigability is well-studied, but a fundamental question is still open: why are some individuals more likely than others to be reached via short, friend-of-a-friend, communication chains? 
In this article we answer the question above by proposing a novel centrality metric called the {\em potential gain,} which, in an informal sense, quantifies the easiness at which a target node can be reached. 
We define two variants of the potential gain, called the {\em geometric} and the {\em exponential potential gain,} and present fast algorithms to compute them. 
The geometric and the potential gain are the first instances of a novel class of {\em composite centrality metrics,} i.e., centrality metrics which combine the {\em popularity} of a node in $\Gr$ with its {\em similarity} to all other nodes. 
As shown in previous studies, popularity and similarity are two main criteria which regulate the way humans seek for information in large networks such as Wikipedia. 
We give a formal proof that the potential gain of a node is always equivalent to the product of its degree centrality (which captures popularity) and its Katz centrality (which captures similarity).
\end{abstract}



\maketitle


\section{Introduction}
\label{sec:introduction}
Centrality metrics \cite{lu2016vital} provide a ubiquitous Network Science tool for the identification of the ``important'' nodes in a graph.
They have been widely applied in a range of domains such as early detection of epidemic outbreaks \cite{chung2009distributing}, viral marketing \cite{leskovec2007dynamics}, trust assessment in virtual communities \cite{AgresteMFPP15}, preventing catastrophic outage in power grids \cite{albert2004structural} and analyzing heterogeneous networks \cite{Agreste-anobii15}.\\
The notion of importance of a node can be defined in a number of ways \cite{newman2010networks,boldi2014axioms,DBLP:journals/netsci/BoldiLV17,DBLP:journals/netsci/Vigna16,DBLP:conf/www/Boldi15}. 
Some centrality metrics define the importance of a node $i$ in a graph $\Gr$ as function of the distance of $i$ to other nodes in $\Gr$: for instance, in  {\em Degree Centrality}, the importance of $i$ is defined as the number of the nodes which are adjacent to $i$, i.e. which are at distance one from $i$. 
Analogously, {\em Closeness Centrality} \cite{newman2010networks} classifies as important those nodes which are few hops away from any other node in $\Gr$.\\
Another class of centrality metrics looks at walk/path structures in $\Gr$: a walk is a sequence of adjacent nodes;its length is defined as the number of edges it contains; a path is a walk without repeated edges and the shortest path connecting two nodes is also called its {\em geodesic path.}
For instance, the {\em Betweenness  Centrality} \cite{newman2010networks} of $i$ is the ratio of the number $g_{jl}(i)$ of geodesic paths from any node $j$ to any node $l$ which pass through the node $i$ to the number $g_{jl}$ of geodesic paths running from $j$ to $l$ and, thus, nodes with largest betweenness  centrality scores are those which intercept most of the geodesic paths in $\Gr$.

A further popular metric is {\em Katz Centrality Score} \cite{katz1953new}, which is understood as the weighted number of walks terminating in $i$: here, the weighting factor is inversely related to walk length and, thus, long (resp., short) walks have a small (resp., large) weight.\\
For a suitable choice of the weighting factor, the Katz centrality score converges to the {\em Eigenvector Centrality} \cite{benzi2014matrix,boldi2014axioms} or the popular {\em PageRank} \cite{brin1998anatomy,boldi2014axioms}.

To the best of our knowledge, however, there is no previous work in which the centrality of a node is closely related to the notion of {\em navigability}: roughly speaking, we say that $\Gr$ is navigable if it is possible to successfully route a message to any node $i$ in $\Gr$ via a short chain of intermediary nodes, regardless of the node $j$ which generates the message.\\
Navigability is one of the most important features for a broad range of natural and artificial systems which have the transportation of information (e.g. a computer network) or the trade of goods (e.g. a road network) as their primary purpose. In general, if the topology of the  graph $\Gr$ underlying the above mentioned systems would be perfectly specified, then any source node $i$ could discover all shortest paths starting from (or terminating in) $i$ and it could make use of the discovered paths to efficiently route messages.\\
In practice, nodes in $\Gr$ are often able to efficiently route messages even if they do not have a global view of the topology of $\Gr$, and this has encouraged many researchers to seek a better understanding of why graphs arising in real applications are navigable.
Early studies on graph navigability were inspired by the seminal work of Travers and Milgram \cite{travers1967small} on the ``small world'' property.\\
In a celebrated experiment, random-chosen Nebraska residents were asked to send a booklet to a complete stranger in Boston. Selected individuals were required to forward the booklet to any of their acquaintances whom they deemed likely to know the recipient or at least might know people who did. In some cases, the booklet actually reached the target recipient by means, on average, of 5.2 intermediate contacts, thus suggesting an intriguing feature of human societies: in large, even planetary-scale, social networks, pairs of individuals are connected through {\em shorts chains of intermediaries} and ordinary people are able to uncover these chains \cite{dodds2003experimental,kleinberg2000small,GoMuWa09,LeHo08}.

Several empirical studies have verified the small-world phenomenon in  diverse domains such as metabolic and biological networks \cite{jeong2000large}, the Web graph \cite{broder2000graph}, collaboration networks among scientists \cite{newman2001structure} as well as social networks \cite{dodds2003experimental,watts1998collective}.

So far, centrality metrics and navigability have been investigated in parallel, yet their research tracks are disconnected. 
Thus, an important (and still unanswered) direction of inquiry is the introduction of centrality metrics that are related to the {\em navigability} of a node, i.e., the ease at which it is possible to reach a target node $i$ regardless of the node $j$ chosen as source node.

In this article we tackle the questions above by extending previous work 
by Fenner {\em et al.} \cite{fenner2008modelling} to the realm of social networks. 
The main output of our research is an index, called the {\em potential gain}, which ranks nodes in a network on the basis of their ability to find a target.

The potential gain of a node $i$ depends on the number of walks $w_k(j,i)$ of length $k$ that connect $i$ with any other node $j.$ 
The underlying idea is that, for a fixed $k,$ the larger $w_k(j, i),$ the higher the chance that $j$ will reach $i,$ regardless of the specific navigation strategy.
In the computation of the potential gain, we take the small-world phenomenon as axiomatic: we consider an agent that starts from $j$ and it looks for short walks to reach $i$. 

We observe that the value a walk has for the agent will decreases with its length $k$ and there is a threshold length beyond which the agent has to abandon that walk.
To formalize the intuition above, we introduce a weighting factor $\phi(k)$ which monotonically decreases with $k$ to penalize long walks.

\noindent
We have developed two variants of the potential gain of \cite{fenner2008modelling}, namely:

\begin{itemize}
	\item the {\em geometric potential gain,} in which $\phi(k)$ decays as $\delta^k$, where $\delta$ is a parameter ranging between $0$ and the inverse of the spectral radius $\lambda_1$ of $\Gr$\footnote{The spectral radius of $\Gr$ is defined as the largest eigenvalue of the adjacency matrix of $\Gr$.}, and
	
	\item the {\em exponential potential gain,} in which $\phi(k)$ decays in exponential fashion.
\end{itemize} 

Both the geometric and exponential gain of $i$ can be thought as the product of one index (Degree Centrality) related to the {\em popularity} of $i$ and another (Katz Centrality score, for the geometric potential gain, and Communicability Index \cite{benzi2013total,estrada2005subgraph} for the exponential potential gain) which reflects the degree of {\em similarity} of $i$ with all other nodes in the network. 
In this sense, the geometric and the exponential potential gain are {\em composite centrality metrics}, i.e., they constitute a novel class of centrality metrics which combine popularity and similarity to rank nodes in graphs. The combination of popularity and similarity has proven to closely resemble the way humans navigate large social networks \cite{csimcsek2008navigating} or attempt to locate information in large information networks such as Wikipedia \cite{west2009wikispeedia,west2012automatic,helic2013models}.

Our formalisation applies the {\em Neuman series expansion} \cite{horn2013matrix} to efficiently but accurately approximate both the geometric and exponential gain.
Both theoretical and experimental analysis show that our approach is appropriate for accurately computing the geometric and exponential potential gain in large real-life graphs consisting of millions of nodes and edges, even with modest hardware resources.

We validated our approach on three large datasets: \textsc{Facebook} (a graph of friendships among Facebook users), \textsc{DBLP} (a graph describing scientific collaboration among researchers in Computer Science) and \textsc{YouTube} (a graph mapping friendship relationships among YouTube users). 
The experimental results will be in the full version of this article.

\COMMENT
\noindent
The main findings of our study can be summarized as follows:

\begin{enumerate}
\item The amount of time needed to compute the geometric or the exponential potential gain {\em does not depend} on the number of nodes/edges of a graph; instead, it depends on the spectral radius $\lambda_1$: 
the larger $\lambda_1$, the better connected the graph and, thus, the larger the number of walks needed to get a good approximation of the geometric/exponential potential gain.

\item For small values of $\delta$, the geometric potential gain is highly correlated with Degree Centrality, while for large values of $\delta$ it displays a strong and positive correlation with Eigenvector Centrality.

\item In the case of the geometric potential gain, walks of small length (i.e., up to  around ten) are sufficient to obtain a good approximation. 
In contrast, to compute the exponential potential gain our algorithm needed to construct longer random walks, in some cases up to ten times longer than those required for the computation of the geometric potential gain.

\item As a consequence of the above point, the geometric potential gain seems to be the most efficient solution for large graphs. 
\end{enumerate} 

This article is organized as follows: in Section \ref{sec:background} we provide basic definitions that will be used throughout the paper. 
In Section \ref{sec:related-works} we review related work.
Section \ref{sec:network-navigability} introduces the geometric and exponential potential gain and illustrates their main properties.
In Section \ref{sec:methods} we discuss how to efficiently calculate the geometric and exponential potential gain, while Section \ref{sec:experiments} details the experiments we have performed. 
Finally, in Section \ref{sec:conc} we draw our conclusions.
\ENDCOMMENT

\section{Background}
\label{sec:background}

In this section we introduce some basic terminology for graphs that will be largely used throughout this article.

Let a graph $\mathcal{G}$ be an ordered pair $\Gr = \langle N, E \rangle$ where \textit{N} is a set of {\em nodes}, here also called {\em vertices}, and $E = \{\langle i, j \rangle: i,j \in V\}$ is the set of {\em edges.} 
As usual, $\Gr$ is {\em undirected} if edges are unordered pairs of nodes and {\em directed} otherwise.
In this article we will consider only undirected graphs. 

Also, let $n = \vert V \vert$ be the number of nodes, $m = \vert E \vert$ the number of edges of $\Gr$. 
For any given node \textit{i} its neighborhood $\N(i)$ is the set of nodes directly connected to it; its {\em degree} $d_i$ is the number of edges incident onto it, i.e., $d_i = \vert \N(i) \vert$.

A {\em walk} of length $k$ (with  $k$ a non-negative integer) is a sequence of nodes $\langle i_0, i_1, \ldots, i_k\rangle$ such that consecutive nodes are directly connected: $\langle i_{\ell}, i_{\ell+1} \rangle\in E$ for $\ell \in [0..k-1].$ 
Also, we use the term {\em path} for walks that do not have repeated vertices. 
A walk will be {\em closed} if it starts and ends at the same node.


We will represent graphs by their associated {\em adjacency matrix,} $\A,$ defined as usual with $a_{ij}=1$ if $\langle i,j\rangle \in E$ and 0 otherwise. 
Sometimes we may slightly simplify notation with $a_{ij} = \A_{ij}.$ 

The adjacency matrix provides a compact formalism to describe many graph properties: for instance, the matrix $\A^2$ where $a^2_{ij} = \sum_{k = 1}^{n} a_{ik}a_{kj}$, gives the number of walks of length two going from $i$ to $j$. 
Inductively, for any positive integer $m$, the matrix $\A^m$ will give the number of closed (resp., distinct) walks of length $m$ between any two nodes $i$ and $j$ if $i = j$ (resp., if $i \neq j$) \cite{cvetkovic1997eigenspaces}.

It is a well-know fact that the adjacency matrix of any undirected graph is {\em symmetric} and, hence, all its eigenvalues $\lambda_1 \geq \lambda_2 \geq \ldots \geq \lambda_n$ are real. 
The largest eigenvalue $\lambda_1$ of $\A$ is also called its {\em principal eigenvalue} or {\em spectral radius} of $\Gr$.
Moreover, the corresponding eigenvectors $\vv_1, \ldots, \vv_n$ will form an orthonormal basis in $\mathbb{R}^n$ \cite{strang1993introduction}.
Eigenpairs $\langle \lambda_i, \vv_i \rangle$ are formed by the eigenvalue $\lambda_i$ and the corresponding eigenvector $\mathbf{v}_i$.


\section{A model of network navigability}
\label{sec:network-navigability}

In this section we introduce our new centrality metrics, called the geometric and exponential potential gain.\\
As we will see, they share a common physical interpretation which is based on the notion of {\em graph navigability}: roughly speaking, we say that a graph $\Gr$ is navigable if, for any target node $i$ in $\Gr$, it is possible to reach $i$ via short paths/walks, independently of the node $j$ (called {\em the source}) from which we choose to start exploring $\Gr$ from.  \\
In the light of previous research on graph navigability, we informally define the {\em navigability score} of a node $i$ as a measure of the ``easiness'' with which it is possible to reach $i$ independently of the source node $j$.
In this way, the navigability score of a node can be interpreted as a {\em centrality metric}.

To define the navigability score we borrow some ideas from previous work by Fenner {\em et al.} \cite{fenner2008modelling}, who formulated the problem of identifying a ``good'' page $p$ from which a user should start exploring the Web. 
A page $p$ is classified as a good starting point if it satisfies the following criteria: {\em (1) it is relevant}, i.e. the content of $p$ closely matches user’s information goals, {\em (2) the page $p$ is central}, i.e., the distance of $p$ to other Web pages in the Web graph is as low as possible and {\em (3) the page $p$  is connected}, in the sense that $p$ is able to reach a maximum number of other pages via its outlinks.\\
A key difference between the approach of Fenner {\em et al.} \cite{fenner2008modelling} and the current one is that they defined the navigability score for $i$ as the ability of $i$ of acting as the {\em source node} for reaching all the other nodes. 
In our setting, instead, we think of the node $i$ as the {\em target node} we wish to reach.\\
So, let us fix a source node $j$ and a target node $i$ and provide an estimate $\tau(j, i)$ of how ``easy'' it will be for $i$ to be reached if we choose $j$ as source node. 
Intuitively, the larger the number of walks from $j$ to $i$, the easier it is for $i$ to be reached from $j$; in addition, we assume that the task of exploring a graph is costly and such cost increases as the length of the walks/paths we use for exploration purposes increases. 
Therefore, shorter walks should be preferred to longer ones. \\ By combining the requirements above, we obtain:

\begin{equation}\label{eqn:tau-ceoff}
\tau(j, i) = \sum_{k=1}^{+\infty}\phi(k)\cdot w_k(j,i)
\end{equation}

\noindent 
here $w_k(j,i)$ is the number of walks of length $k$ going from $j$ to $i$ and the non-increasing function $\phi(k)$ acts as penalty for longer walks. 
If we sum over all possible source nodes $j$, we obtain a global centrality index $p(i)$ for $i$:

\begin{equation}\label{eqn:potential-as-sum}
p(i) = \sum_{j \in N} \tau(j,i).
\end{equation}

In analogy to Fenner {\em et al.} \cite{fenner2008modelling,levene2004navigating}, we will call $p(i)$  the {\em potential gain} of $i$. \\
Depending on the choice of the penalty function $\phi(\cdot)$ we obtain two variants of the potential gain, namely the geometric and the exponential potential gain (see Section \ref{sub:geometric-exponential-potantial-gain}).

\COMMENT
In Section \ref{sub:relation-to-centrality} we will compare the geometric and the exponential potential gain with other, well-known, centrality metrics from the literature.\\
In Section \ref{sub:relation-geometric-exponential} we investigate the relationship between the geometric and the exponential potential gain.
Finally, Section \ref{sub:fast-calculation-geom-expon} outlines our approach to calculating the geometric and exponential potential gain.
\ENDCOMMENT

\subsection{The geometric and exponential potential gain}
\label{sub:geometric-exponential-potantial-gain}

Given the above specifications, we first define the potential gain in matrix notation.
For the base case, consider walks of length $k$=1, i.e., direct connections. 
Only the neighbours of a node $i$ will contribute to the potential gain of $i,$ which leads to the trivial conclusion that, at $k= 1$, nodes with the largest degree are also those ones with the largest potential gain. 

We define the vector $\mathbf{p}$ such that $\mathbf{p}_i = p(i)$ for every node $i$:

\begin{equation}
\label{eqn:potential-matrix}
\mathbf{p} = \phi(1)\cdot \mathbf{A} \times \mathbf{1}.
\end{equation}

If we include walks of length two, then we have to consider the squared adjacency matrix $\mathbf{A}^2$. 
So, we add a contribution $\phi(2)\cdot \mathbf{A}^2 \times \mathbf{1}$ to the potential gain.

By induction, nodes capable of reaching from $i$ through walks of length up to $k$ provide a contribution to the potential gain equal to $\phi(k)\cdot \mathbf{A}^k \times \mathbf{1}$. 
By summing over all possible values of $k$ we get to the following expression for $\mathbf{p}$:

\begin{align*}
\label{eqn:potential-gain-matrix}
\mathbf{p} &= \phi(1)\mathbf{A} \times \mathbf{1} +
\phi(2)\mathbf{A}^2 \times \mathbf{1} + \ldots + \phi(k)\mathbf{A}^k
\times \mathbf{1} + \ldots\\ 
& =\sum_{k=1}^{+\infty} \left(\phi(k)\mathbf{A}^k\times
\mathbf{1}\right) = 
\left(\sum_{k=1}^{+\infty}\phi(k)\mathbf{A}^k\right) \times \mathbf{1}
\end{align*}

To attenuate the effect of the walks' length, we will consider two weighting functions, namely:

\begin{enumerate}
\item {\em Geometric:} $\phi(k) = \delta^{k-1}$ with $\delta \in (0,1)$. 
So we define the {\em geometric potential gain,} $\mathbf{g}$:

\begin{equation}
\label{eqn:geometric-potential-gain}
\mathbf{g} = \left(\mathbf{A} + \delta \mathbf{A}^2 + \ldots + \delta^{k-1} \mathbf{A}^k + \ldots \right) \times\mathbf{1}
\end{equation} 

\item {\em Exponential:} $\phi(k) = \frac{1}{(k-1)!}$. 
So we define the {\em exponential potential gain,} $\mathbf{e}$:

\begin{equation}
\label{eqn:exponential-potential-gain}
\mathbf{e} = \left(\mathbf{A} + \mathbf{A}^2 + \ldots + \frac{1}{\left(k -1 \right) !} \mathbf{A}^k + \ldots\right) \times \mathbf{1}
\end{equation} 

\end{enumerate}

\section{Potential Gain as centrality}
\label{sub:relation-to-centrality}

The geometric and the exponential potential gain introduced above yield a {\em ranking} of network nodes and, therefore, it is instructive to compare them with popular centrality metrics.
Recall that we defined the {\em spectral radius} $\lambda_1$ of $\mathbf{A}$ as the largest eigenvalue of $\mathbf{A}$.

As for the geometric potential gain, if we let $\delta < \lambda_1^{-1}$, the following expansion holds:

\begin{align*}
\mathbf{g} &= \left(\mathbf{A} + \delta \mathbf{A}^2 + \ldots + \delta^{k-1} \mathbf{A}^k + \ldots \right) \times\mathbf{1}  \\
&=\mathbf{A} \times \left(\mathbf{I} + \delta \mathbf{A} + \ldots + \delta^{k-1} \mathbf{A}^{k-1} + \ldots \right) \times\mathbf{1} \\
&=\mathbf{A} \times \left(\mathbf{I} - \delta \mathbf{A}\right)^{-1} \times\mathbf{1}
\label{eqn:geometric-potential-series}
\end{align*}

\noindent 
in which we make use of the {\em Neuman series} \cite{horn2013matrix}
  
\begin{equation}
\label{eqn:neuman-series}
\left(\mathbf{I} +  \ldots + \delta^{k-1} \mathbf{A}^{k-1} + \ldots \right) = \left(\mathbf{I} - \delta \mathbf{A}\right)^{-1}.
\end{equation}

At this point, the term $\left(\mathbf{I} - \delta \mathbf{A}\right)^{-1} \times \mathbf{1}$ is exactly the {\em Katz centrality score} \cite{katz1953new,leicht2006vertex}, a popular centrality metric that defines the importance of a node as a function of its similarity with other nodes in $\Gr.$
Hence, we can say that the geometric potential gain combines two kind of contributions: {\em popularity,} as captured by node degree, and {\em similarity} as captured by Katz's similarity score. 


It is also instructive to consider what happens for extreme values of $\delta$: if $\delta \to 0$, then the geometric potential gain tends to $\mathbf{A} \times \mathbf{1}$, i.e., it coincides with degree.
In contrast, if $\delta \to \frac{1}{\lambda_1}$, then the Katz centrality score converges to {\em eigenvector centrality} \cite{benzi2014matrix}, another popular metric adopted in Network Science. 
Boldi et al. \cite{boldi2014axioms, DBLP:journals/netsci/BoldiLV17, DBLP:journals/netsci/Vigna16, DBLP:conf/www/Boldi15} show that the Katz Centrality score is also strictly related to the PageRank. 
More specifically, the PageRank vector $\mathbf{p}$ coincides with the Katz Centrality score provided that the adjacency matrix $\mathbf{A}$ is replaced by its row-normalized version $\overline{\mathbf{A}}$:
\begin{equation}
\mathbf{p} = \left(1 - \alpha\right) \sum_{k=0}^{+\infty} \alpha^i \overline{\mathbf{A}}^i \times \mathbf{1}
\end{equation} 
Here, the parameter $\alpha$ is the so-called PageRank {\em damping factor}.
Let us now concentrate on the exponential potential gain. We rewrite Equation \ref{eqn:exponential-potential-gain} as follows:

\begin{align*}
\mathbf{e} &= \left(\mathbf{A} + \mathbf{A}^2 + \ldots +
\frac{1}{\left(k -1 \right) !} \mathbf{A}^k + \ldots\right) \times
\mathbf{1}\nonumber\\ 
 &=
\mathbf{A} \times \left(\mathbf{I} + \mathbf{A} + \ldots + \frac{1}{k
!} \mathbf{A}^k + \ldots\right) \times \mathbf{1}\\
& =
\mathbf{A} \times \exp(\mathbf{A}) \times \mathbf{1}\nonumber
\label{eqn:exponential-potential-series}
\end{align*}

\noindent
where $\exp(\mathbf{A}) = \sum_{k=1}^{+\infty} \frac{1}{k!} \mathbf{A}^k$ is the exponential of $\mathbf{A}$ \cite{higham2008functions}. 

The exponential of a matrix has been used to introduce other centrality scores such as {\em communicability} or {\em subgraph centrality} \cite{estrada2012physics,benzi2014matrix}. 

Specifically, $\exp\left(\mathbf{A}\right)_{ij}$ measures how easy is to send a unit of flow from a node $i$ to a node $j$ and vice versa.
Such a parameter is known as {\em communicability} and it can be regarded as a measure of similarity between a pair of nodes. 
Communicability has been successfully used to discover communities in networks \cite{estrada2012physics}. 
The product $\exp(\mathbf{A}) \times \mathbf{1}$ yields a centrality metric which defines the importance of a node as function of its ability to communicate with all other nodes in the network.
In turn, the diagonal entry $\exp\left(\mathbf{A}\right)_{ii}$ of the matrix exponential defines a further centrality metric called {\em subgraph centrality} \cite{estrada2005subgraph}. 
As a result of the rewriting above, we clearly see exponential potential gain as dependent on two factors: popularity of $i$ (i.e., its degree) and similarity of $i$ with all other nodes in the network.

The computation of the geometric (resp., exponential) potential gain for {\em all nodes} in $\Gr$ needs the specification of the full adjacency matrix $\mathbf{A}$; in this sense, the geometric and the exponential potential gain should be considered as {\em global centrality metrics}, on par with the Katz centrality score and Subgraph centrality.

\COMMENT
\subsection{The relation between the Geometric and the Exponential Potential Gain}
\label{sub:relation-geometric-exponential}

In this section we present some guidelines on how to choose the $\delta$ factor discussed in the previous section. 

A straightforward choice would be to set $\delta = (2\lambda_1)^{-1}$ \cite{katz1953new} or, in analogy with the Google PageRank damping factor, $\delta =0.85\lambda_1^{-1}$ \cite{benzi2013total}. On the other hand, Foster et al. \cite{foster2001faster} suggested the following:

$$
\delta = \frac{1}{\norm{\mathbf{A}}_{\infty} + 1}
$$

\noindent 
where $\norm{\mathbf{A}}_{\infty} = \max_{1 \leq i \leq n} \sum_{j= 1}^{n} \vert \mathbf{A}_{ij}\vert$.

It is instructive to investigate the existence of a  {\em crossover point} $\delta^{c}$, i.e. to discover a value of $\delta$ at which the geometric and the exponential gain of a node $i$ coincide. To this end, we provide the following result.

\begin{theorem}
Let $\Gr$ be a graph with adjacency matrix $\mathbf{A}$ and eigenvalues $\lambda_1 \geq \lambda_2 \geq \dots \geq \lambda_n$. 
For each node $i$ and for $\delta \in (0, \lambda_1^{-1})$, the geometric and the exponential gains of $i$ coincide if and only if one of the following holds:

\begin{enumerate}
	\item $\lambda_i = 0$, or

	\item $\delta = \delta^{c} = \frac{e^{\lambda_i} - 1}{\lambda_ie^{\lambda_i}}$, provided that $\delta^{c} < \lambda_1^{-1}$.
\end{enumerate}
\end{theorem}

\begin{proof}
Recall that for sufficiently large values of $k$, we can approximate the geometric and the exponential gain defined in Equations \ref{eqn:geometric-potential-series} and \ref{eqn:exponential-potential-series} as follows:

\[
\mathbf{g} = \mathbf{A} \times \left(\mathbf{I} - \delta \mathbf{A}\right)^{-1} \times\mathbf{1} \quad \mbox{and} \quad \mathbf{e} = \mathbf{A} \times \exp(\mathbf{A}) \times \mathbf{1}
\]

Recall that $\mathbf{A}$ is a square and symmetric matrix.
Thus, it admits the following eigendecomposition, 

\[
\mathbf{A} = \mathbf{D}^{-1} \times \mathbf{\Lambda} \times \mathbf{D}
\]

where $\mathbf{\Lambda}$ is a diagonal matrix storing the eigenvalues $\lambda_1, \lambda_2, \ldots, \lambda_n$ of $\mathbf{A}$ and $\mathbf{D}$ is an orthonormal matrix, whose columns coincide with the eigenvectors $\mathbf{u}_1, \mathbf{u}_2, \ldots \mathbf{u}_n$ of $\mathbf{A}$.

Now recall \cite{higham2008functions} that, for any function $f$, the matrix $f(\mathbf{A})$ is still diagonalizable and, for any eigenvalue $\lambda_i$ of $\mathbf{A}$ we have that $f(\lambda_i)$ is an eigenvalue of $f(\mathbf{A})$. 
In addition, the matrices $\mathbf{A}$ and $f(\mathbf{A})$ share the same eigenvectors so we have $f(\mathbf{A}) = \mathbf{D}^{-1} \times f(\mathbf{\Lambda}) \times \mathbf{D}$.

Let us consider now the application of the two functions $f_1(x) = \frac{x}{1 - \delta x}$ and $f_2(x) = xe^x$ to matrix $\mathbf{A}$. The eigenvalues of the matrix $f_1(\mathbf{A}) = \mathbf{A} \times \left(\mathbf{I} - \delta \mathbf{A}\right)^{-1}$ are

\begin{equation}
\label{eqn:lambda_gpg}
\frac{\lambda_1}{1 - \delta \lambda_1}, \frac{\lambda_2}{1 - \delta \lambda_2}, \ldots, \frac{\lambda_n}{1 - \delta \lambda_n}
\end{equation}

whereas the eigenvalues of the matrix  $f_2(\mathbf{A}) = \mathbf{A} \times \exp(\mathbf{A})$ are 

\begin{equation}
\label{eqn:lambda_epg}
\lambda_1e^{\lambda_1}, \lambda_2e^{\lambda_2}, \ldots, \lambda_ne^{\lambda_n}.
\end{equation}

Let us introduce $\mathbf{\Lambda}_g$ and $\mathbf{\Lambda}_e$, the diagonal matrices storing the eigenvalues of the matrices $\mathbf{A} \times \exp(\mathbf{A})$ and $\mathbf{A} \times \left(\mathbf{I} - \delta \mathbf{A}\right)^{-1}$, respectively. 
Let us now compute the difference between the geometric and potential gain:

\begin{dmath}
\label{eqn:difference-geometric-potential}
\mathbf{g} - \mathbf{e} = \mathbf{A} \times \left(\mathbf{I} - \delta \mathbf{A}\right)^{-1} \times\mathbf{1} - \mathbf{A} \times \exp(\mathbf{A}) \times \mathbf{1} = \mathbf{D}^{-1} \times \mathbf{\Lambda_g} \times \mathbf{D} \times \mathbf{1} -\mathbf{D}^{-1} \times \mathbf{\Lambda_e} \times \mathbf{D} \times \mathbf{1} =
\left(\mathbf{D}^{-1} \times \left(\mathbf{\Lambda_g} - \mathbf{\Lambda_e}\right) \times \mathbf{D}\right) \times \mathbf{1}
\end{dmath}

We focus on the $i$-th component of vector $\mathbf{g} - \mathbf{e}$ and observe that its value $\Delta_i$ is given as:

\[
\Delta_i = \left(\lambda_i e^{\lambda_i} - \frac{\lambda_i}{1 - \delta \lambda_i}\right) \mathbf{u}_i^T\mathbf{u}_i =  \left(\lambda_i e^{\lambda_i} - \frac{\lambda_i}{1 - \delta \lambda_i}\right)
\]

Here we used the fact that eigenvectors of $\mathbf{A}$ form an orthonormal basis.
If we assume that $\lambda_ i \neq 0$, then $\Delta_i = 0$ if and only if:

\begin{equation}
\label{eqn:delta-crossover}
\delta = \frac{e^{\lambda_i} - 1}{\lambda_ie^{\lambda_i}}
\end{equation}

\noindent
which completes the proof.
\end{proof}

\subsection{Calculation of Geometric and Exponential Potential Gains}
\label{sub:fast-calculation-geom-expon}

In this section we present our algorithm for the computation of the geometric potential and exponential potential gain. 

\blue{Our algorithm can be implemented in few lines of code in any high-level programming language, since it applies the expansion series provided in Equations \ref{eqn:geometric-potential-series} and \ref{eqn:exponential-potential-series}.
Our approach provides insight on how the walk length $k$ affects the calculation of the geometric (resp., exponential) potential gain: indeed, if we stop the expansion of Equation \ref{eqn:geometric-potential-series} (resp. Eq. \ref{eqn:exponential-potential-series}) after the first $k$ terms, then, we would only include the walks up to length $k$ in the calculation of the geometric (resp., exponential) potential gain.}

Let us consider the computational complexity of our solution.
\blue{We begin with the} geometric potential gain and assume that we stop expanding the Neumann series after generating walks of length $k^{\star}$.
In such a case, it is easy to see that cost will be in $O(k^{\star}\vert E\vert)$. 
\blue{In fact, for any $j$ such that $1 < j < k^{\star}$, let us set 
$\mathbf{y}_j = \delta^{j-1}\mathbf{A}^{j} \times \mathbf{1}$ and 
suppose that we have stored the sequence $\mathcal{Y} = \{\mathbf{y}_1, \mathbf{y}_2, \ldots, \mathbf{y}_{j-1}\}$, with $\mathbf{y}_1 = \mathbf{1}$, $\mathbf{y}_2 = \mathbf{A} \times \mathbf{1}$.}   

\noindent
Hence, the following recurrence holds:
 
\begin{dmath}
\label{eqn:updates-geometric-potential-gain}
\mathbf{y}_j
= 
\delta^{j-1} \mathbf{A}^{j}\times \mathbf{1} 
= 
\left(\delta \mathbf{A}\right) \times \left(\delta^{j-2} \mathbf{A}^{j-1}\times \mathbf{1}\right)
= 
\left(\delta \mathbf{A}\right) \times \mathbf{y}_{j-1}
\end{dmath}

The last equality states that any term $\mathbf{y}_j$ can be calculated as the product of a sparse matrix ($\delta \mathbf{A}$) by a vector ($\mathbf{y}_{j-1}$), already computed in the previous iteration. 
Such an operation takes $O(\vert E \vert)$ steps which, in the case of sparse networks, is $O(n)$.

Similarly, given that $\mathbf{g}$ can be expressed as $\mathbf{g} = \sum_{j=0}^{k^{\star}} \mathbf{y}_j$, we conclude that the cost required to compute the geometric potential gain amounts is $O(k^{\star}n)$.
As for space complexity, the cost for computing $\mathbf{g}$ is $O(\vert E\vert)$.

The computation of the geometric potential gain requires to fix $\delta$ beforehand, which, in turn, requires to fix an approximation of the spectral radius $\lambda_1$. 
The literature on the estimation of $\lambda_1$ provides some bounds on it \cite{das2004some,stevanovic2015spectral} but, available upper bounds are often not tight and, thus, uninformative; therefore, an alternate way to approximate $\lambda_1$ is to rely on algorithms such as the {\em Power Iteration Method} \cite{heath2002scientific}. 
On the other hand, if we target very large graphs, {\em sampling techniques} seem the best option \cite{han2017closed}.

Analogous results for both time and space complexity hold for the computation of the exponential potential gain as we show next. 
Define a sequence $\mathcal{Z} = \{\mathbf{z}_i\}$ recursively as follows:

\begin{align*}
\mathbf{z}_1 = \mathbf{1} \\
\mathbf{z}_2 = \mathbf{A} \times \mathbf{1}\\
\dots\\
\mathbf{z}_i = \frac{1}{i -2}\mathbf{A}\times\mathbf{z}_{i-1}
\end{align*}

Therefore, any term $\mathbf{z}_j$ can be calculated as the product of a sparse matrix ($\mathbf{A}$) by a vector ($\mathbf{z}_{j-1}$), which has been already computed in the previous iteration. 
Such an operation takes $O(\vert E \vert)$.

Given that $\mathbf{e}$ can be expressed as $\mathbf{e} \simeq \sum_{j=0}^{k^{\star}} \mathbf{z}_j$, we can conclude that the (worst-case) time complexity for the calculation of the exponential potential gain is $O(k^{\star}\vert E \vert)$; similarly the space complexity is $O(\vert E \vert)$, hence for sparse graphs both time and space complexity reduce to $O(n).$

\ENDCOMMENT



\section{Conclusions}
\label{sec:conc}

We have introduced the potential gain, an index to rank nodes in graphs that captures the ability of a node to act as a target point for navigation within the network.
We have defined two variants of the potential gain, the geometric and exponential potential gain. 
We then proposed two iterative algorithms that compute the geometric and exponential potential gain and proved their convergence.
We evaluated the scalability of our algorithms on three real large datasets. 

We have discovered connections between the geometric potential gain and other, well-known, centrality metrics; GPG provides a new, mixed global-local centrality measure.
Indeed, the PG as a centrality index has several merits:

\begin{itemize}
    \item it unifies Katz and Communicability  into a single framework;
    
    \item in its definition in terms of the PG it allows us to provide novel and efficient approximations of these indices;
    
    \item it provides an instance of a novel class of composite indices, in this case  DC*Katz, and
    
    \item as each vertex has clear visibility of its neighbours, the realisation that PG is a combination of local (Degree) and global (Katz) centrality makes complete sense, in our opinion.
\end{itemize}

It is also possible that these results will open the door to a new interpretation of social phenomena related to Travers-Milgram's ``small world'' experiment \cite{travers1967small}.

One question that could be discussed at this point is which of the two new measures could be considered the best analysis tool large networks. 
Early experimental results indicate different rates of convergence but no clear ``winner.''

From a computational standpoint, the geometric potential gain is clearly superior. 
So, for the analysis of very large networks and/or modest hardware resources it is the navigability score of choice.
One practical difference however remains.
The exponential potential gain is parameter-free and can be applied directly.
On the other hand, the geometric potential gain is parametric in $\delta$ thus it requires a careful tuning of the algorithm.

\COMMENT
The importance of $\delta$ is also underscored by the fact that for values close or equal to $1/\lambda_1$ we observed the scores for the geometric potential gain, the Exponential potential gain and for Katz centrality to fall into some sort of alignment.
\ENDCOMMENT

Another topic for future work is investigation on the relationship between network robustness and network navigability. 
To this end, we intend to design an experiment in which graph nodes are ranked on the basis of their geometric/exponential potential gain and then are progressively removed from the graph.
Basic properties about graph topology, such as the number and size of connected components shall be re-evaluated upon node deletion.
We also plan to study how adding edges can increase the geometric/exponential potential gain of a target group of nodes.

\bibliographystyle{ACM-Reference-Format}
\bibliography{potential_gain+boldi}

\end{document}